\documentclass[final,3p,times]{elsarticle}
\usepackage[T1]{fontenc}
\usepackage{amsfonts}
\usepackage{amssymb}
\usepackage{amsmath}
\usepackage{amsthm}
\usepackage{graphics}
\usepackage{graphicx}
\usepackage{xcolor}
\usepackage{braket}
\usepackage[colorlinks = true, linkcolor = blue, urlcolor = blue, citecolor = blue]{hyperref}
\newtheorem{theorem}{Theorem}

\newtheorem{definition}{Definition}

\newtheorem{lemma}{Lemma}
\newtheorem{proposition}{Proposition}

\makeatletter
\def\ps@pprintTitle{   
\let\@oddhead\@empty   
\let\@evenhead\@empty   
\let\@oddfoot\@empty   
\let\@evenfoot\@oddfoot
}
\makeatother

\begin{document}
\begin{frontmatter}
\title{Unextendible product bases, bound entangled states, and the range criterion}

\author[1]{Pratapaditya Bej}
\ead{pratap6906@gmail.com}
\address[1]{Department of Physics and Centre for Astroparticle Physics and Space Science, Bose Institute,\\ EN-80, Sector V, Bidhannagar, Kolkata 700091, India}

\author[2]{Saronath Halder}
\ead{saronath.halder@gmail.com}
\address[2]{Quantum Information and Computation Group, Harish-Chandra Research Institute, HBNI,\\ Chhatnag Road, Jhunsi, Prayagraj (Allahabad) 211 019, India}

\begin{abstract}
An unextendible product basis (UPB) is a set of orthogonal product states which span a subspace of a given Hilbert space while the complementary subspace contains no product state. These product bases are useful to produce bound entangled (BE) states. In this work we consider reducible and irreducible UPBs of maximum size, which can produce BE states of minimum rank. From a reducible UPB, it is possible to eliminate one or more states locally, keeping the post-measurement states orthogonal. On the other hand, for an irreducible UPB, the above is not possible. Particularly, the UPBs of the present size are important as they might be useful to produce BE states, having ranks of the widest variety, which satisfy the range criterion. Here we talk about such BE states. We also provide other types of BE states and analyze certain properties of the states. Some of the present BE states are associated with the tile structures. Furthermore, we provide different UPBs corresponding to the present BE states of minimum rank and discuss important properties of the UPBs.   
\end{abstract}

\end{frontmatter}

\section{Introduction}\label{sec1}
Quantum entanglement \cite{Horodecki09-2, Guhne09} is regarded as a valuable resource for its application in quantum cryptography \cite{Ekert91}, superdense coding \cite{Bennett92}, quantum teleportation \cite{Bennett93}, etc. In these protocols, to get the maximum advantage, we often talk about the usage of pure maximally entangled states in the bipartite settings. But in a practical situation, it is never possible to avoid noise which results mixed entangled states. Therefore, one has to follow certain procedures to distill pure state entanglement from mixed entangled states \cite{Bennett96-1, Bennett96-2, Bennett96-3} by the allowed class of operations, i.e., local operations and classical communication (LOCC). We note here that it is not always possible to extract entanglement in pure form by applying distillation procedures to any mixed entangled state. This is because there are bound entangled (BE) states \cite{Horodecki97, Horodecki98}. Clearly, such entangled states are mixed states from which it is not possible to get any pure entangled state via LOCC. This is true even though many copies of a BE state are available.

We now consider that a bipartite quantum system is given and the associating Hilbert space is $\mathcal{H}$ = $\mathbb{C}^{d_1}\otimes\mathbb{C}^{d_2}$. If it is possible to compute negative eigenvalues by taking partial transpose of the state of that system then it guarantees that the state of the system is entangled (or inseparable) \cite{Peres96}. This is also a necessary condition for inseparability when $d_1d_2\leq6$ \cite{Horodecki96}. But when $d_1d_2>6$, it is not always easy to conclude if the state is separable or entangled. This is due to some bipartite entangled states which remain positive under partial transpose (PPT) \cite{Horodecki97}. It is known that if a bipartite entangled state remains PPT then the state is a BE state \cite{Horodecki98}. In this context, we mention that bipartite BE states with negative partial transpose may exist (conjectured) but it is an open problem till date \cite{DiVincenzo00, Dur00, Pankowski10}. However, in this work we restrict ourselves to the BE states which are PPT (for a multipartite system, PPT across every bipartition). Here we say these states as PPT entangled states or simply the BE states. 
   
A common technique, to detect the entanglement of a bipartite PPT entangled state, is to find a suitable indecomposable map which is positive but not completely positive, for example, the Choi map \cite{Choi75}. Applying such a map to a bipartite PPT entangled state, it is possible to compute negative eigenvalue(s) which guarantees the inseparability of the given  state. Nevertheless, for any bipartite PPT entangled state, there is no universal way to form a suitable positive but not completely positive indecomposable map which detect the entanglement of the state. Due to such a complexity, it is important to find new classes of BE states. We also remember here that starting from a positive but not completely positive indecomposable map, it is possible to constitute Hermitian operators which can witness bound entanglement \cite{Lewenstein00, Lewenstein01}. 

In Ref.~\cite{Horodecki97}, to prove the inseparability of some PPT entangled states, the range criterion was introduced. According to this criterion, for any bipartite separable state $\rho$, it is possible to find bipartite product states $\{\ket{a_i}\ket{b_i}\}_i$ which span the support of $\rho$ while the states $\{\ket{a_i}\ket{{b_i}^\ast}\}_i$ span the support of $\rho^{\mathbb{T}}$. Here $\mathbb{T}$ stands for partial transpose (PT) with respect to the second subsystem and $\ast$ stands for the complex conjugation operation in a particular basis with respect to which PT is taken. If the state of a given quantum system, violates the range criterion, then the state is an entangled state. Obviously, for any Hilbert space the full-rank states always satisfy the range criterion. Interestingly, there also exist low-rank PPT entangled states which can satisfy the range criterion \cite{Bandyopadhyay05, Halder19-1}. For PPT entangled states, it is anyway difficult to detect the entanglement. Increasing this complexity, we now have some low-rank PPT entangled states which can satisfy the range criterion. Clearly, detection of entanglement within these states, is a very difficult task. Hence, it is important to explore different techniques to constitute such states having the lowest rank to full rank.

BE states can be constructed efficiently by using the notion of unextendible product basis (UPB) which was introduced in Ref.~\cite{Bennett99}. For a given Hilbert space $\mathcal{H}$ = $\mathbb{C}^{d_1}\otimes\mathbb{C}^{d_2}\otimes\cdots\otimes\mathbb{C}^{d_m}$, we consider a set $\mathcal{S}$ of pure orthogonal product states $\{\ket{\theta_i}\}_{i=1}^N$ (here the states are fully separable). Suppose, these states span a subspace $\mathcal{H}_S$ of the Hilbert space $\mathcal{H}$. Now, the complementary subspace $\mathcal{H}_E$ is a completely entangled subspace, i.e., it does not contain any product state if the states of the set $\mathcal{S}$ form an unextendible product basis.\footnote{In a given Hilbert space, for a complete orthogonal product basis also, it is not possible to find another product state which is orthogonal to all the states of the considered basis. But in this case the dimension of the completely entangled subspace is zero which is not desired here. Therefore, we exclude those bases from the set of unextendible product bases.} The BE state due to the UPB is given as $\varrho$ = $(1/N^\prime)(\mathbb{I}-\sum_{i=1}^N\ket{\theta_i}\bra{\theta_i})$, $N^\prime$ is the dimension of $\mathcal{H}_E$ and $\mathbb{I}$ is the identity operator acting on $\mathcal{H}$. It is known that for a bipartite system if one of the subsystems is of dimension two then it is not possible to construct any UPB in the associating Hilbert space \cite{Bennett99, Divincenzo03}. 

UPBs are studied not only for PPT entangled states. It is also known that the states within a UPB cannot be perfectly distinguished by LOCC \cite{Bennett99, Divincenzo03, Rinaldis04}. Therefore, UPBs show the phenomenon--quantum nonlocality without entanglement \cite{Bennett99-1}. Again, the UPB generated PPT entangled states are supported in the completely entangled subspaces (normalized projector onto the complementary subspace $\mathcal{H}_E$). This implies that the UPB generated PPT entangled states are the so-called edge states\footnote{For a bipartite system, an edge state is a PPT entangled state from which no separable projector of rank-$1$ can be subtracted (with a small proportion) preserving both, positivity and PPT \cite{Lewenstein00, Lewenstein01}.} which violate the range criterion in an extreme manner \cite{Lewenstein00, Lewenstein01}. Now, starting from these edge states, sometimes it is possible to construct PPT entangled states which satisfy the range criterion (for bipartite systems see Refs.~\cite{Bandyopadhyay05, Halder19-1}). 

The present work focuses on the UPBs of maximum size which can produce BE states of minimum rank. In fact, there are bipartite BE states of minimum rank, which can be produced from tile structures. We also discuss about the (in)distinguishability properties of UPBs under LOCC. There are many articles where the authors have discussed about the minimum cardinality\footnote{The cardinality (or the size) of a UPB is nothing but the number of states present in a UPB.} of UPBs (see Ref.~\cite{Chen15} and the references therein). This is important as it gives the idea of small sets of LOCC indistinguishable product states. But the UPBs of maximum size are also important as they help to generate a class of noisy bound entangled states having ranks of the widest variety. These states are noisy, in a sense that, such a state can be written as a convex combination of a separable state (noise) and an edge state. Here the separable states are produced by taking the convex combination of the pure orthogonal product states from the given UPB. Furthermore, the noisy BE states may satisfy the range criterion. 

The study of noisy bound entangled states also gives insight regarding the robustness of entanglement within the edge states. In this context, we mention that in Ref.~\cite{Sentis18}, the authors have presented a class of BE states, entanglement within which is robust against noise and these states are fit for experimental verification. For any BE state, it is difficult to say how to use the state as resource in quantum information processing protocols. However, in certain scenarios, the researchers have shown the usage of particular BE states, for example, secure key distillation \cite{Horodecki05, Horodecki08, Horodecki09-1}, quantum steering \cite{Moroder14}, quantum nonlocality \cite{Vertesi14, Yu17}, quantum metrology \cite{Toth18}, etc. These examples clearly depict the importance of studying BE states.

We now provide the main contributions of the present work: (i) We discuss about a class of UPBs of maximum size, the states of which are with real normalizing coefficients. Such a UPB has the property that if the stopper is removed then the rest of the UPB can be extended to a complete orthogonal product basis. We show that these UPBs can be applied to produce noisy bound entangled states, having ranks of the widest variety, which satisfy the range criterion. This is true for both bipartite and multipartite systems. (ii) We construct both reducible and irreducible UPBs of maximum size in this work. We prove that in $\mathbb{C}^3\otimes\mathbb{C}^3$ and in $\mathbb{C}^2\otimes\mathbb{C}^2\otimes\mathbb{C}^2$, all UPBs are irreducible. Moreover, we also prove that in $\mathbb{C}^{d_1}\otimes\mathbb{C}^{d_2}$, $d_1, d_2 > 9$, $d_1, d_2 \geq 3$, for any rank-4 BE state, it is possible to consider that the state is due to a reducible UPB. (iii) We show that looking at an edge state, it may not be possible to conclude if the edge state is due to a reducible UPB or due to an irreducible UPB. This is true for both bipartite and multipartite systems (except $\mathbb{C}^3\otimes\mathbb{C}^3$ and $\mathbb{C}^2\otimes\mathbb{C}^2\otimes\mathbb{C}^2$ systems). We also characterize two-qutrit noisy bound entangled states. (iv) We provide a class of UPBs of maximum size for any bipartite systems. The states of which are with real coefficients. Moreover, such a UPB has the property that if the stopper state is removed then the rest of the UPB can be extended to a complete orthogonal product basis. This type of UPBs can be applied to produce noisy bound entangled states, having ranks of the widest variety, which satisfy the range criterion. We also show that these UPBs are reducible UPBs (except when the bipartite system is a $\mathbb{C}^3\otimes\mathbb{C}^3$ system). (v) For multipartite systems, we further construct noisy bound entangled states which satisfy the range criterion. Such states can have rank ranging from five to full rank. 

The rest of the paper is arranged in the following way: In Sec.~\ref{sec2}, we provide a few important concepts. The results for the bipartite systems are given in Sec.~\ref{sec3} and the same for the multipartite systems are given in Sec.~\ref{sec4}. Finally, the conclusion is drawn in Sec.~\ref{sec5}.

\section{Preliminaries}\label{sec2}
In this section, we present a few important concepts. These concepts are briefly introduced in the introduction. However, here we provide those concepts in a greater details as they are frequently used in the rest of the paper. 

\begin{definition}
(Unextendible product basis \cite{Bennett99, Divincenzo03}) Given a set of pure orthogonal product states, if these states span only a subspace of the whole Hilbert space, in such a way that the complementary subspace contains no product state, then the set is an unextendible product basis.
\end{definition}

Here we consider both reducible and irreducible UPBs of maximum size. Definitions of which can be given as the following. But before we give the definitions, it is important to understand the orthogonality-preserving LOCC. Suppose, a set of orthogonal states are provided to distinguish. The process of discrimination can be a multi-round protocol. If after each round, the remaining states, to be distinguished, remain pairwise orthogonal to each other then such a LOCC protocol is orthogonality preserving.  

\begin{definition}
(Reducible and irreducible UPB) Given a UPB, if it is not possible to eliminate any state from the UPB via orthogonality-preserving LOCC then the UPB is an irreducible UPB. Again, if it is possible to eliminate one or more states from the UPB via orthogonality-preserving LOCC then the UPB is a reducible UPB. 
\end{definition}

We note here that in case of a multipartite system, the product states within a UPB must be separable across every bipartition. The incomplete product bases of the above kind (the UPBs), are useful to produce bound entangled states which are positive under partial transpose. In a multipartite system such incomplete product bases produce bound entangled states which are positive under partial transpose across every bipartition.

\begin{definition}
(Edge state \cite{Lewenstein00, Lewenstein01}) A bipartite edge state $\delta_{edge}$, is positive under partial transpose entangled state, such that there exists no product state $\ket{a}\ket{b}$ and $\epsilon>0$, for which $\delta-\epsilon\ket{a}\bra{a}\otimes\ket{b}\bra{b}$ is a positive operator and it is also positive under partial transpose.
\end{definition}

In this work, we particularly, use the edge states which are due to UPBs. These states are supported in the entangled subspaces and they violate the range criterion in an extreme manner. The multipartite edge states which are due to UPBs, are supported in the completely entangled subspaces. Notice that to define a multipartite edge state, we have to use a pure state which is separable across every bipartition in the above definition.

\begin{definition}
(Noisy bound entangled states \cite{Halder19-1}) A bound entangled state which is positive under partial transpose, is said to be noisy if it can be written as a convex combination of an edge state and a separable state (noise).
\end{definition}

In the present work we mostly mention the real UPBs, the states of which are normalized with real coefficients, in order to construct the noisy bound entangled states. Thus, the noisy bound entangled states or the edge states which are considered in the paper are mostly invariant under partial transpose (see Refs.~\cite{Bandyopadhyay05, Halder19-1}). In fact, the noisy bound entangled states are produced here by using particular separable states, which are generated by taking the convex combinations of different states, picked from a UPB.

\begin{lemma}
(The range Criterion \cite{Horodecki97}) Given a bipartite separable state $\rho$, it is always possible to find bipartite product states $\{\ket{a_i}\ket{b_i}\}_i$ which span the support of $\rho$ while the states $\{\ket{a_i}\ket{{b_i}^\ast}\}_i$ span the support of $\rho^{\mathbb{T}}$. Here $\mathbb{T}$ stands for partial transpose with respect to the second subsystem and $\ast$ stands for the complex conjugation operation in a particular basis with respect to which partial transpose is taken. This is known as the range criterion. If the state of a given quantum system, violates the range criterion, then it is guaranteed that the state is an entangled state.
\end{lemma}

Remember that satisfaction of the range criterion is a necessary condition for separability while a violation of this criterion is a sufficient condition for inseparability (entanglement).

\section{Bipartite systems}\label{sec3}
Given a bipartite quantum system, we suppose that the associating Hilbert space is $\mathcal{H}$ = $\mathbb{C}^{d_1}\otimes\mathbb{C}^{d_2}$; $d_1,d_2\geq3$. As of now, we use $d_1\otimes d_2$ instead of $\mathbb{C}^{d_1}\otimes\mathbb{C}^{d_2}$. We then think of a UPB, the states of which span a subspace $\mathcal{H}_S$ of the whole Hilbert space $\mathcal{H}$, where $\dim{\mathcal{H}_S} < d_1d_2$. It is known that the normalized projector onto the subspace $\mathcal{H}_E$ (which is an orthogonal subspace to $\mathcal{H}_S$), is a BE state \cite{Bennett99}. The dimension of $\mathcal{H}_E$ is ($d_1d_2-\dim{\mathcal{H}_S}$) and it is an entangled subspace. The minimum nonzero dimension of $\mathcal{H}_E$ in this case can be four. This is due to the fact that the minimum rank of a BE state is four. We mention here that in Ref.~\cite{Horodecki03-1}, it was shown that rank-$2$ bipartite BE state does not exist and later in Ref.~\cite{Chen08}, it was shown that rank-$3$ bipartite entangled states are distillable. Evidently, the maximum size of a UPB is ($d_1d_2-4$), which can produce BE states of minimum rank. The UPBs of the above size may have a very interesting application: They generate the edge states of rank-$4$ and starting from these edge states it may possible to construct noisy BE states, having rank-$5$ to full rank, which satisfy the range criterion. In fact, this is the widest variety of ranks, a noisy BE state can have, because the BE states of minimum rank (rank-$4$) are all edge states \cite{Chen11, Chen13}. Nevertheless, to realize the construction, we use Theorem $1$ of Ref.~\cite{Halder19-1} and apply it on the real UPBs (the states of which are with real normalizing coefficients) of size ($d_1d_2-4$). If one considers real UPBs then the constructions become easier. This is due to the following facts: (i) the edge states or the noisy BE states which are produced due to real UPBs, are invariant under PT, (ii) so, to prove that a noisy BE state satisfy the range criterion, it is sufficient to show that there exist sufficient pure product states in the range of the considered BE state and these product states span the whole range. We consider, these UPBs also have the property that if the stopper state (or simply the stopper) is removed from a UPB, then the rest of the set can be extended to a full basis by adding other pairwise orthogonal product states. Hence, by taking different convex combinations of these UPB generated edge states and the stopper states, one can get noisy BE states which satisfy the range criterion \cite{Halder19-1}. More precisely, using Theorem $1$ of Ref.~\cite{Halder19-1}, from the real UPBs of the above kind, we can construct noisy BE states of rank-$(d_1d_2-N+1)$ to full rank, which satisfy the range criterion, where $N$ is the size of the real UPBs. Here $N$ is equal to $(d_1d_2-4)$, which implies the existence of noisy BE states of rank-$5$ to full rank, which satisfy the range criterion. Next, we summarize the above discussion into a proposition:

\begin{proposition}\label{prop1}
Consider a real UPB of size ($d_1d_2-4$) in $d_1\otimes d_2$ and assume that if the stopper is removed then the remaining states of the UPB, can be extended to a full basis. Such a UPB is useful to produce noisy BE states, having ranks of the widest variety, which satisfy the range criterion.
\end{proposition}

It is important to remember that in $3\otimes3$, the stopper state is not special. If any state from a $3\otimes3$ UPB is removed then the rest of the UPB can be extended to a complete orthogonal product basis. This is due to two facts: (i) in $3\otimes3$ only one size of UPB is possible and that is five \cite{Divincenzo03}, (ii) in bipartite systems, any four orthogonal states can always be extended to a complete orthogonal product basis \cite{Divincenzo03}. However, it may not be true for all the UPBs, constructed in the present work. This is why we have used the feature if the stopper is removed then the rest of the UPB can be extended to a complete orthogonal product basis.

The $3\otimes3$ Hilbert space is a special case as there is no reducible UPB in this Hilbert space. It is also important to note here that in $3\otimes3$, in the entangled subspace due to a UPB, there is only one BE state while it is possible to get more than one BE states in the entangled subspace due to a UPB in $d\otimes d$, $d$ is odd \cite{Halder19-3}. Remember that the UPBs, which are constructed in Ref.~\cite{Halder19-3}, are used to present a new class of bound entangled states which are the boundary points of the convex compact set of states with positive partial transpose. But those bound entangled states are supported in the entangled subspaces, and therefore, those states violate the range criterion in an extreme manner. So, this is completely opposite to the motivation of finding noisy bound entangled states, having ranks of the widest variety, which satisfy the range criterion. In the present paper, we have given an algorithm to generate some rank-4 bound entangled states which are supported in the entangled subspaces (the algorithm is given in a later portion). But the present rank-$4$ bound entangled states are due to reducible UPBs of maximum size, moreover, the states of the UPBs are normalized with real coefficients. These UPBs also have the property that if the stopper is removed then the rest of the UPB is extendible to a complete orthogonal product basis. As a result of which using the present rank-$4$ bound entangled states, it is possible to construct noisy bound entangled states, having ranks of the widest variety, which satisfy the range criterion. Also note that the UPBs, which are constructed in Ref.~\cite{Halder19-3}, are not of maximum size, in fact, they may not be reducible UPBs. Moreover, the states of the UPBs of Ref.~\cite{Halder19-3}, may contain complex normalizing coefficients. So, it is not known how to generate range criterion satisfying bound entangled states, having ranks of the widest variety, from the UPBs of Ref.~\cite{Halder19-3}. However, we now present the following proposition:

\begin{proposition}\label{prop2}
In $3\otimes3$, all rank-$4$ BE states are due to irreducible UPBs.
\end{proposition}

\begin{proof}
It is known that in $3\otimes 3$, all rank-$4$ BE states are due to UPBs \cite{Chen11}. Again, there is no reducible UPB in $3\otimes 3$. This is because of the following facts: In $3\otimes3$, there are UPBs of cardinality five (maximum cardinality in this case) only \cite{Divincenzo03}. Now, if it is possible to eliminate one or more states from such a UPB by orthogonality-preserving LOCC then the remaining states (four orthogonal product states or less) can always be distinguished by LOCC \cite{Divincenzo03}. Thus, a UPB is being distinguishable by LOCC but it is known to be impossible \cite{Bennett99, Divincenzo03, Rinaldis04}. In this way, all rank-$4$ BE states in $3\otimes3$ are due to irreducible UPBs.
\end{proof}

Apart from $3\otimes3$, in all other bipartite Hilbert spaces $d_1\otimes d_2$, $d_1d_2>9$, $d_1, d_2\geq3$, there are reducible UPBs. This can be realized very easily: We consider the Tiles UPB in $3\otimes3$ \cite{Bennett99}, the states of which are given by- 
\begin{equation}\label{eq1}
\begin{array}{c}
\ket{0}\ket{0-1},~\ket{0-1}\ket{2},~\ket{2}\ket{1-2},~\ket{1-2}\ket{0},~\ket{0+1+2}\ket{0+1+2},
\end{array}
\end{equation}
where $\ket{c_1\pm c_2\pm\cdots\pm c_n}\equiv (\ket{c_1}\pm\ket{c_2}\pm\cdots\pm\ket{c_n})$, ignoring the normalizing coefficients. [We use such notations in many places of this paper.] Now, starting from the above UPB one can keep adding simple pure orthogonal product states to get new UPBs in higher dimensions. For example, one can add the states $\ket{0}\ket{3}$, $\ket{1}\ket{3}$, $\ket{2}\ket{3}$ with the above UPB to get a new UPB in $3\otimes4$. But notice that from the new UPB it is possible to eliminate the states $\ket{0}\ket{3}$, $\ket{1}\ket{3}$, $\ket{2}\ket{3}$ via orthogonality-preserving LOCC. For such reducible UPBs, we present the following proposition:

\begin{proposition}\label{prop3}
In $d_1\otimes d_2$, $d_1d_2>9$, $d_1, d_2\geq3$, for any rank-$4$ BE state, it is possible to consider that the state is due to a reducible UPB. 
\end{proposition}

\begin{proof}
All rank-$4$ BE states are supported in $3\otimes3$ subspaces \cite{Chen13}. Moreover, two-qutrit rank-$4$ BE states are due to UPBs \cite{Chen11}. So, such a state can be considered due to a UPB of that subspace. If it is the case then the remaining dimension $(d_1d_2-9)$ can be filled with simple product states which can be eliminated by orthogonality-preserving LOCC. Thus, we arrive to the above proposition.
\end{proof}

We now move to the discussion of UPBs and BE states in $3\otimes 4$. We have provided a tile structure in Fig.~\ref{fig1}. From the tile structure, it is possible to construct a BE state which is supported in $3\otimes3$. This can be done without considering the UPB. We first think of a complete orthogonal product basis (COPB), from which a UPB can be produced. In order to construct the UPB, one can think about the stopper state $\ket{0+1+2}\ket{0+1+2+3}$. Obviously, the stopper state is not orthogonal to certain states of the COPB. Those states of the COPB can be found from the tile structure of Fig.~\ref{fig1}. Following Ref.~\cite{Halder19-3}, we say the states as `missing states'. These missing states are given by-
\begin{equation}\label{eq2}
\begin{array}{c}
\ket{0}\ket{0+1+2},~\ket{0+1}\ket{3},~\ket{2}\ket{1+2+3},~\ket{1+2}\ket{0},~\ket{1}\ket{1+2}.
\end{array}
\end{equation}
Taking the linear combination of the above five states, it is possible to construct four entangled states which are orthogonal to the stopper state. For example, consider the states (not normalized): $\ket{0}\ket{0+1+2}-\ket{2}\ket{1+2+3}$, $\ket{0+1}\ket{3}-\ket{1+2}\ket{0}$, $2(\ket{0}\ket{0+1+2}+\ket{2}\ket{1+2+3})-3(\ket{0+1}\ket{3}+\ket{1+2}\ket{0})$, $\ket{0}\ket{0+1+2}+\ket{0+1}\ket{3}+\ket{2}\ket{1+2+3}+\ket{1+2}\ket{0}-5\ket{1}\ket{1+2}$. These states are entangled states and they are orthogonal to the product states of the UPB. These entangled states span a four-dimensional entangled subspace which is a part of a bigger $3\otimes3$ subspace of the entire $3\otimes4$ Hilbert space. The $3\otimes3$ subspace is spanned by $\{\ket{0},\ket{1},\ket{2}\}\otimes\{\ket{0},\ket{1+2},\ket{3}\}$. The normalized projector onto the four-dimensional entangled subspace is a BE state of rank-$4$ because this state has similar structure like the BE state due to the UPB, given in Eq.~(\ref{eq1}). The present rank-$4$ BE state of $3\otimes4$ can be due to a reducible UPB as mentioned in Proposition \ref{prop3}. Such a UPB is given as the following:

\begin{equation}\label{eq3}
\begin{array}{c}
\ket{0}\ket{1-2},~ \ket{0}(2\ket{0}-\ket{1+2}),~\ket{1}\ket{1-2},~\ket{2}\ket{1-2},~\ket{2}(\ket{1+2}-2\ket{3}),~ \ket{1-2}\ket{0},~\ket{0-1}\ket{3},\\[1 ex]
\ket{0+1+2}\ket{0+1+2+3}.
\end{array}
\end{equation}
Notice that the party who is holding the four-level system (say, Bob) can perform an orthogonality-preserving projective measurement to eliminate certain states. Corresponding projection operators are $\mathcal{P}_1=\ket{1-2}\bra{1-2}$ and $\mathcal{P}_2=\mathbb{I}-\ket{1-2}\bra{1-2}$, where $\mathbb{I}$ is an identity operator acting on a four-level quantum system. The proof of the unextendibility of the above states comes from the following facts: (i) the above states are orthogonal to each other and they span an eight-dimensional subspace of the $3\otimes4$ Hilbert space, (ii) the normalized projector onto the rest four-dimensional subspace is a BE state (as mentioned earlier) which implies that the range of the BE state must not has sufficient orthogonal product states to span the four-dimensional subspace, so, the above states cannot be extended to a COPB anyway, (iii) if there are one or more orthogonal product states in the four-dimensional subspace then adding those states with the above states can give a UPB, (iv) but BE states cannot have rank $<4$. Thus, the states of the above equation forms a UPB.

\begin{figure}[h]
\centering
\includegraphics[scale=0.22]{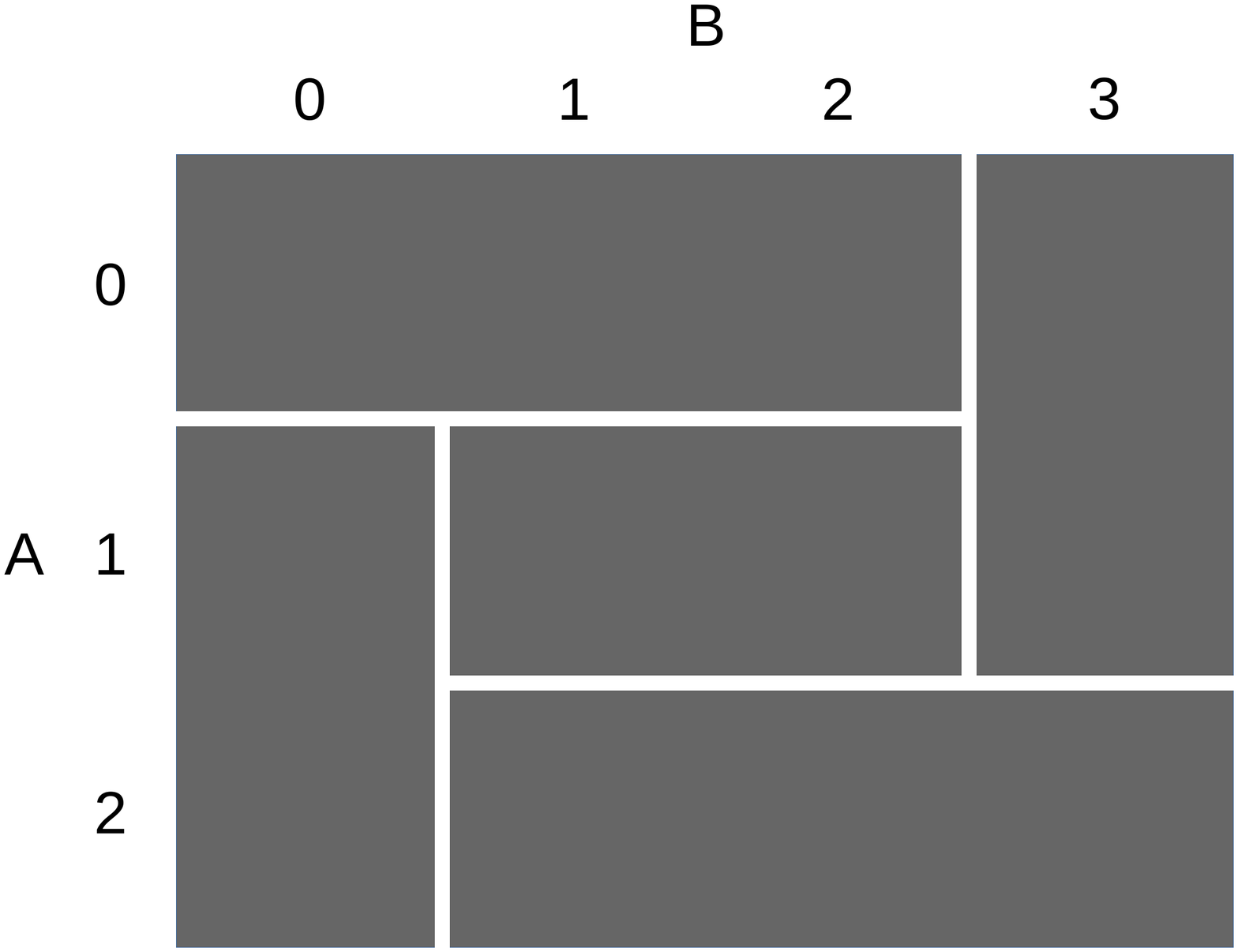}
\caption{A tile structure in $3\otimes 4$ Hilbert space.}\label{fig1}
\end{figure}

The rank-$4$ BE state which we have constructed here, can also be due to an irreducible UPB in $3\otimes4$. Such an irreducible UPB is given as the following:
\begin{equation}\label{eq4}
\begin{array}{c}
\ket{0}\ket{0-1},~ \ket{0}(\ket{0+1}-2\ket{2}),~\ket{1}\ket{1-2},~\ket{2}\ket{2-3},~\ket{2}(2\ket{1}-\ket{2+3}),~ \ket{1-2}\ket{0},~\ket{0-1}\ket{3},\\[1 ex]
\ket{0+1+2}\ket{0+1+2+3}.
\end{array}
\end{equation}
Notice that the four orthogonal entangled states which are constructed after Eq.~(\ref{eq2}), span a four-dimensional Hilbert space. Moreover, these entangled states are pairwise orthogonal to the states, given in Eq.~(\ref{eq3}) and also to the states, given in Eq.~(\ref{eq4}). This implies that the entangled states which are constructed after Eq.~(\ref{eq2}) together with the eight orthogonal product states of Eq.~(\ref{eq3}) or with the eight orthogonal product states of Eq.~(\ref{eq4}), can span the whole $3\otimes4$ Hilbert space. So, this can be true only when the product states of Eq.~(\ref{eq3}) and that of Eq.~(\ref{eq4}) span the same subspace.

Following the same argument as given for the UPB of Eq.~(\ref{eq3}), it is possible to show that the above states form a UPB. Again, the proof of irreducibility comes from the fact that no party is able to begin with a nontrivial orthogonality-preserving measurement (in this regard see the Refs.~\cite{Groisman01, Walgate02}). For Alice (who is holding the three-level quantum system), any measurement can be described by a set of $3\times3$ positive matrices. It can be shown that if Alice wants to eliminate any state from the above UPB keeping the remaining states orthogonal then the positive matrices are proportional to $3\times3$ identity matrix. So, the measurement becomes trivial. For Bob (who is holding the four-level quantum system), any measurement can be described by a set of $4\times4$ positive matrices. It can be shown that if Bob wants to eliminate any state from the above UPB keeping the remaining states orthogonal then the positive matrices are proportional to $4\times4$ identity matrix. So, the measurement becomes trivial. For more details see the Appendix section. We next summarize the above discussion into a proposition:

\begin{proposition}\label{prop4}
In $d_1\otimes d_2$, $d_1d_2>9$, $d_1,d_2\geq3$, consider an edge state which is due to a reducible UPB. The same edge state can also be due to an irreducible UPB. Thus, looking at the edge state, it may not be possible to conclude if the state is due to a reducible UPB or it is due to an irreducible UPB.
\end{proposition}

Notice that both the UPBs (reducible and irreducible) in $3\otimes4$ can be considered as real UPBs. They also have the property that if the stopper is removed then the rest of a UPB can be extended to a complete basis by adding other orthogonal product states. Hence, starting from the rank-$4$ BE state, constructed here, one can construct noisy BE states, having rank-$5$ to full rank, which satisfy the range criterion. The details of the construction procedures can be found in Refs.~\cite{Bandyopadhyay05, Halder19-1}. Particularly, in Ref.~\cite{Bandyopadhyay05}, starting from a fixed real UPB in $3\otimes3$, it was shown how a class of noisy BE states having rank $\geq5$, can be constructed, which satisfy the range criterion. However, a general statement regarding this can be given in the following way:

\begin{proposition}\label{prop5}
The two-qutrit BE states, having rank-$5$ to full rank, satisfy the range criterion if it has a form, $\sigma = q\delta_{sep}+(1-q)\delta_{edge}$, $q$ is a small positive number, $\delta_{sep}=\sum_ip_i\ket{\varphi_i}\bra{\varphi_i}$, $\ket{\varphi_i}$s are product states and $\delta_{edge}=\sum_ip_i^\prime\ket{\varphi_i^\prime}\bra{\varphi_i^\prime}$ is a rank-$4$ edge state, $\forall i$ $\ket{\varphi_i}$, $\ket{\varphi_i^\prime}$ are pairwise orthogonal states with real normalizing coefficients.
\end{proposition}

\begin{proof}
It is known that the two-qutrit rank-$4$ bound entangled states are due to UPBs \cite{Chen11}. So, they are supported in the entangled subspaces. The state $\delta_{sep}$ is supported in a subspace which is orthogonal to an entangled subspace. $\delta_{sep}$ has rank $r$, where $1\leq r \leq5$. The remaining subspace of the given two-qutrit Hilbert space, where the state $\sigma$ is not supported, has the dimension $\leq4$. There, one can think about a set of remaining orthogonal product states $\mathbb{S}$ = $\{\ket{\varphi_i}\}_i$ which are not in the convex combination of $\delta_{sep}$ and those product states $\{\ket{\varphi_i}\}_i$ span the remaining subspace of dimension $\leq4$ (where the state $\sigma$ is not supported). Again, it is also known that for a bipartite system, any four pure orthogonal product states (or less number of such states) can always be extended to a COPB \cite{Divincenzo03}. So, the orthogonal product states of the set $\mathbb{S}$ can be extended to a COPB. We denote the set of product states, which is required for the extension of $\mathbb{S}$ to a COPB, by $\mathbb{S}^\prime$. Remember that the product states of $\mathbb{S}$ and that of $\mathbb{S}^\prime$ are orthogonal to each other. Clearly, the range of $\sigma$ is being spanned by the product states of the set $\mathbb{S}^\prime$. Thus, it is sufficient to prove that $\sigma$ satisfies the range criterion. 
\end{proof}

The above proposition is particularly important because there exist two-qutrit BE states which have rank greater than $4$ and can violate the range criterion \cite{Clarisse06, Ha07, Kye12}. Interestingly, in $3\otimes4$, there is only one type of reducible UPB from the cardinality point of view and this is given in the following proposition:

\begin{proposition}\label{prop6}
In $3\otimes4$, all reducible UPBs are of size eight. 
\end{proposition}

\begin{proof}
To construct a reducible UPB, one has to choose a subspace of the given Hilbert space first and then can think about a UPB in that subspace. Thereafter, with that subspace UPB, one has to add a few simple product states. So, the size of the reducible UPB in a given Hilbert space is dependent on the different sizes of the UPBs of the considered subspace. In case of $3\otimes4$ Hilbert space, one can only think about a two-qutrit subspace because $3\otimes3$ is the lowest bipartite dimension where a UPB can exist \cite{Divincenzo03}. It is also known that there is only one type of UPBs in $3\otimes3$ and that is the UPBs of size five \cite{Divincenzo03}. So, with such a UPB of size five in the two-qutrit subspace, one can add another three simple pure orthogonal product states to produce a new UPB in $3\otimes4$. The new UPB is of size eight and the simple product states which are not included in the two-qutrit subspace UPB can be eliminated via orthogonality-preserving LOCC. In this way, all reducible UPBs in $3\otimes4$, are of size eight. 
\end{proof}

If we go beyond $3\otimes4$, then there are reducible UPBs of different sizes. Unlike Proposition \ref{prop4}, in $3\otimes4$ if a UPB generated edge state is given, whose rank is greater than $4$, then obviously, the state is due to an irreducible UPB. We again go back to the discussion of noisy BE states and present the following proposition:

\begin{proposition}\label{prop7}
Consider a rank-$5$ BE state which has a form, $\sigma = p\delta_{sep}+(1-p)\delta_{edge}$, $p$ is a small positive number, $\delta_{sep}$ is a separable state and it is orthogonal to $\delta_{edge}$ which is a rank-$4$ edge state. If the state $\sigma$ satisfies the range criterion then it must be supported in a two-qutrit subspace.
\end{proposition}

\begin{proof}
Any Rank-$4$ BE state is supported in an entangled subspace of a two-qutrit Hilbert space \cite{Chen13}. Now, in the above it is considered that $\sigma$ is a rank-$5$ state while $\delta_{edge}$ is a rank-$4$ state and $\delta_{sep}$ is a rank-$1$ separable state which is orthogonal to $\delta_{edge}$. Next, we consider a bigger Hilbert space and we suppose that $\delta_{sep}$ is not supported in the two-qutrit subspace in which $\delta_{edge}$ is. So, the five-dimensional subspace, where $\sigma$ is supported, must have deficit of product states. How exactly the `deficit of product states' is occurring, can be understood in the following way: firstly, the rank-4 BE state is supported in an entangled subspace as mentioned earlier, secondly, we are assuming that the rank-$1$ separable state is orthogonal to the rank-$4$ state and in fact, it is not supported in the two-qutrit subspace where the rank-$4$ BE state is supported. Thus, in the five-dimensional subspace where the state $\sigma$ is supported, must contain only one product state, i.e., the state $\delta_{sep}$. Example of such five-dimensional subspaces can be found in Ref.~\cite{Halder19-1}. Also note that these five-dimensional subspaces are the complementary subspaces of uncompletable product bases (for uncompletable product bases, see Ref.~\cite{Divincenzo03}) and they have deficit of product states. The above explanation implies that $\sigma$ must not satisfy the range criterion. Clearly, if $\sigma$ satisfies the range criterion then $\delta_{sep}$ must belong to the two-qutrit subspace where $\delta_{edge}$ is supported. Hence, it is proved that the rank-$5$ state $\sigma$ must be supported in the two-qutrit subspace if it satisfies the range criterion. 
\end{proof}

In Ref.~\cite{Bandyopadhyay05}, the authors used a particular reducible UPB, to construct low-rank noisy BE states which satisfy the range criterion. It is also possible to explore noisy BE states due to reducible UPBs which are maximally robust in certain directions \cite{Halder19-1}. Furthermore, in the present paper we use rank-$4$ BE states in several places to produce noisy BE states, having rank-$5$ to full rank, which satisfy the range criterion. So, it is important to provide protocols to produce rank-$4$ BE states and reducible UPBs. We first give a protocol to construct rank-$4$ BE states which are supported in different $3\otimes 3$ subspaces of a bigger bipartite Hilbert space ($d_1\otimes d_2$). Thereafter, we construct a class of reducible UPBs in $d_1\otimes d_2$. All these UPBs have the size ($d_1d_2-4$) and they can be considered as real UPBs. Respect to the reducible UPBs of Ref.~\cite{Bandyopadhyay05}, the present UPBs are important because they help to explore different two-qutrit entangled subspaces of a bigger $d_1\otimes d_2$ Hilbert space. For the present purpose, we now consider the tile structure, given in Fig.~\ref{fig2}. 

\begin{figure}[h]
\centering
\includegraphics[scale=0.24]{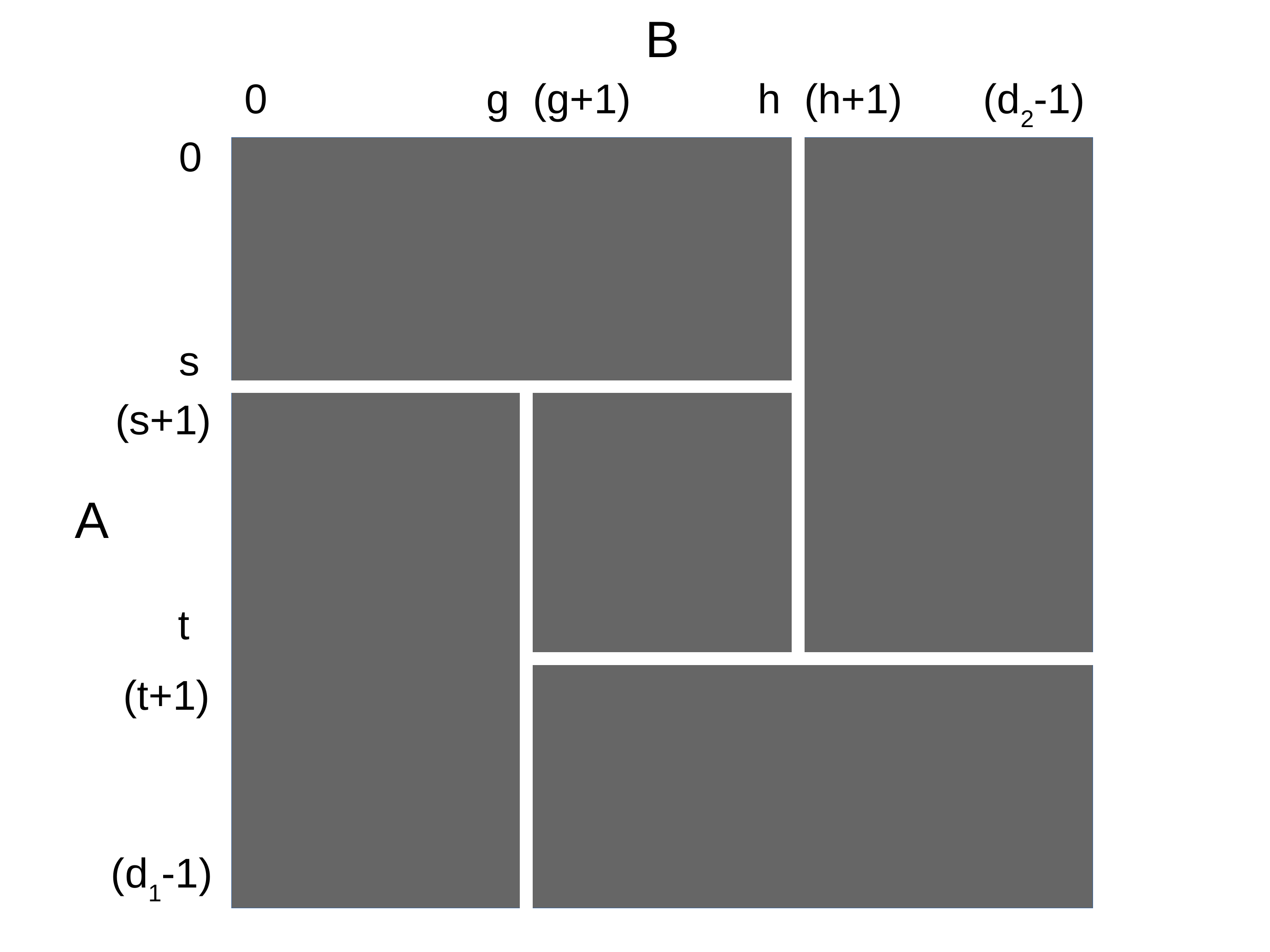}
\caption{A tile structure in $d_1\otimes d_2$ Hilbert space.}\label{fig2}
\end{figure}

Consider the following states first: 
\begin{equation}\label{eq5}
\begin{array}{l}
\ket{\phi_0^a}=\ket{0+1+\cdots+s},~\ket{\phi_1^a}=\ket{(s+1)+(s+2)+\cdots+t},~\ket{\phi_2^a}=\ket{(t+1)+(t+2)+\cdots+(d_1-1)},\\[1 ex] 
\ket{\phi_0^b} = \ket{0+1+\cdots+g},~\ket{\phi_1^b} = \ket{(g+1)+(g+2)+\cdots+h},~\ket{\phi_2^b} = \ket{(h+1)+(h+2)+\cdots+(d_2-1)}.
\end{array}
\end{equation}
In the above we can consider $0\leq s\leq(d_1-3)$, and if $s=0$, then $\ket{\phi_0^a}=\ket{0}$. Again, if $(d_1-1-s)=2$, then $(s+1)=t$, $(t+1)=(d_1-1)$, which implies $\ket{\phi_1^a}=\ket{d_1-2}$ and $\ket{\phi_2^a}=\ket{d_1-1}$. But remember that $t$ is strictly less than $(d_1-1)$. In the similar way, we define $1\leq h\leq(d_2-2)$. If $h=1$ then $g=0$ and $(g+1)=h$ which implies $\ket{\phi_0^b}=\ket{0}$ and $\ket{\phi_1^b}=\ket{1}$. Also remember that $g$ is strictly less than $h$. Again, if $h=(d_2-2)$ then $(h+1)=(d_2-1)$ which implies $\ket{\phi_2^b}=\ket{d_2-1}$. Now, look at the above tile structure and consider any COPB which fits in the tile structure. In order to construct a UPB, we can consider a stopper state $\ket{s}=\ket{0+1+\cdots+d_1}\ket{0+1+\cdots+d_2}$. Due to the use of the stopper state, certain states of the COPB become nonorthogonal to the stopper state. Such states are missing states as defined earlier. The missing states are given by-
\begin{equation}\label{eq6}
\begin{array}{c}
\ket{\psi_1} = \ket{\widetilde{\phi_0^a}}\ket{\widetilde{\phi_0^b}+\widetilde{\phi_1^b}},~\ket{\psi_2} = \ket{\widetilde{\phi_0^a}+\widetilde{\phi_1^a}}\ket{\widetilde{\phi_2^b}},~\ket{\psi_3} = \ket{\widetilde{\phi_2^a}}\ket{\widetilde{\phi_1^b}+\widetilde{\phi_2^b}},~\ket{\psi_4} = \ket{\widetilde{\phi_1^a}+\widetilde{\phi_2^a}}\ket{\widetilde{\phi_0^b}},~\ket{\psi_5} = \ket{\widetilde{\phi_1^a}}\ket{\widetilde{\phi_1^b}},
\end{array}
\end{equation}
where $\ket{\widetilde{\phi_i^a}}$ and $\ket{\widetilde{\phi_i^b}}$ are the normalized versions of $\ket{\phi_i^a}$ and $\ket{\phi_i^b}$ respectively. Using the above product states it is possible to construct four entangled states which are orthogonal to the stopper state. Then the normalized projector onto the subspace, spanned by those entangled states is a rank-$4$ BE state supported in a two-qutrit subspace of $d_1\otimes d_2$ Hilbert space. In fact, these type of rank-$4$ BE states have the similar structure as that of Bennett {\it et al}.~in $3\otimes3$ [the BE state which is due to the UPB, given in Eq.~(\ref{eq1})]. Here the two-qutrit Hilbert space is spanned by $\{\ket{\widetilde{\phi_0^a}}, \ket{\widetilde{\phi_1^a}}, \ket{\widetilde{\phi_2^a}}\}\otimes\{\ket{\widetilde{\phi_0^b}}, \ket{\widetilde{\phi_1^b}}, \ket{\widetilde{\phi_2^b}}\}$. These BE states are special as all these states are extreme points of the convex compact set of states with PPT \cite{Chen11, Halder19-3}. Starting from such states it is possible to construct low-rank BE states which satisfy the range criterion. Moreover, such states can be due to reducible UPBs when $d_1d_2>9$ (by Proposition \ref{prop3}). To realize those UPBs we consider the following:
\begin{equation}\label{eq7}
\begin{array}{c}
\ket{\psi_1^\prime} = \ket{\widetilde{\phi_0^a}}\ket{\widetilde{\phi_0^b}-\widetilde{\phi_1^b}},~
\ket{\psi_2^\prime} = \ket{\widetilde{\phi_0^a}-\widetilde{\phi_1^a}}\ket{\widetilde{\phi_2^b}},~
\ket{\psi_3^\prime} = \ket{\widetilde{\phi_2^a}}\ket{\widetilde{\phi_1^b}-\widetilde{\phi_2^b}},~
\ket{\psi_4^\prime} = \ket{\widetilde{\phi_1^a}-\widetilde{\phi_2^a}}\ket{\widetilde{\phi_0^b}},\\[1 ex]
\ket{s^\prime} = \ket{\widetilde{\phi_0^a}+\widetilde{\phi_1^a}+\widetilde{\phi_2^a}}\ket{\widetilde{\phi_0^b}+\widetilde{\phi_1^b}+\widetilde{\phi_2^b}}.
\end{array}
\end{equation}
Clearly, the above states form a UPB in a two-qutrit subspace as it has a similar form of the UPB by Bennett {\it et al}.~in $3\otimes3$, see Eq.~(\ref{eq1}). Now, to extend that to a UPB in $d_1\otimes d_2$, one can consider simple product states which can be eliminated by some orthogonality-preserving LOCC. In order to realize these product states, one should consider first the states $\ket{\Phi_j^a}$ and $\ket{\Phi_j^b}$ which are orthogonal to $\ket{\widetilde{\phi_i^a}}$ and $\ket{\widetilde{\phi_i^b}}$ respectively $\forall i=0,1,2$. We then consider the following orthogonal states: \{$\ket{\Phi_j^a}\ket{\widetilde{\phi_i^b}}$, $\ket{\widetilde{\phi_i^a}}\ket{\Phi_j^b}$, $\ket{\Phi_j^a}\ket{\Phi_{j^\prime}^b}$\}. These states together with the states of the above equation form a class of reducible UPBs of size ($d_1d_2-4$). Starting from the present rank-$4$ BE states, one can construct the noisy BE states which can satisfy the range criterion. To construct such states one can take the convex combination of a rank-$4$ BE state and a separable state. Again, the separable states are produced by taking the convex combination of pure product states chosen from the present UPBs (for details see Refs.~\cite{Bandyopadhyay05, Halder19-1}). UPBs of maximum size are also constructed in Ref.~\cite{Shi20} but those UPBs are different from the present UPBs. However, we next consider certain multipartite systems and discuss different constructions, properties of those systems.

\section{Multipartite systems}\label{sec4}
In this section we start with the three-qubit systems and prove the following proposition first:

\begin{proposition}\label{prop8}
All UPBs are irreducible UPBs for any three-qubit system. 
\end{proposition}

\begin{proof}
It has been shown that for any three-qubit system, all UPBs are of size four \cite{Bravyi04}. The proof of which is based on the fact that for three qubits, mixed states up to certain ranks are separable \cite{Karnas01}. Furthermore, any set of three or less number of multipartite pure orthogonal fully separable states can be distinguished by LOCC \cite{Divincenzo03}. Hence, we arrive to the above proposition.
\end{proof}

For all other multipartite systems, it is possible to construct reducible UPBs following the same argument as given for the bipartite systems. Remember that the minimum rank of a multipartite BE state can be $4$ and these rank-$4$ multipartite BE states are supported in three-qubit completely entangled subspaces\footnote{A completely entangled subspace does not contain any pure state which is separable across every bipartition \cite{Parthasarathy04}.} (CES) \cite{Chen13}. Using these facts, one can conclude that the maximum size of a multipartite UPB is $(\mathcal{D}-4)$, which produce multipartite BE states of minimum rank, while $\mathcal{D}$ = $d_1d_2\dots d_m$ is the total dimension of a multipartite Hilbert space, $\mathcal{H}$ =  $\mathbb{C}^{d_1}\otimes\mathbb{C}^{d_2}\otimes\cdots\otimes\mathbb{C}^{d_m}$ (as of now, we only use $d_1\otimes d_2\otimes\cdots\otimes d_m$ instead of $\mathbb{C}^{d_1}\otimes\mathbb{C}^{d_2}\otimes\cdots\otimes\mathbb{C}^{d_m}$). Starting from a multipartite real UPB (the states of which are with real normalizing coefficients) of size $(\mathcal{D}-4)$, it is also possible to define a class of multipartite noisy BE states, having ranks of the widest variety, which satisfy the multipartite range criterion. In the following we show explicit constructions of those BE states. Note that in this section we only talk about multipartite systems and thus, by BE states, range criterion, edge states, UPBs etc., we mean these things for multipartite systems only. We mention here that the definition of multipartite edge states and the range criterion for any multipartite system can be given by following the bipartite arguments (for details see Ref.~\cite{Kiem15}). For us, these multipartite definitions become simpler as we stick to real UPBs to construct noisy BE states.

We start with a three-qubit system and construct noisy BE states. Subsequently, we prove that the noisy BE states satisfy the range criterion. Particularly, for three-qubit systems, we consider a result from Ref.~\cite{Divincenzo03}, to prove that the noisy BE states satisfy the range criterion. For the construction, we consider the well know Shift UPB \cite{Bennett99, Divincenzo03} for a three-qubit system. The UPB is given as the following:
\begin{equation}\label{eq8}
\begin{array}{l}
\ket{\Psi_1} = \ket{0}\ket{1}\ket{0-1},~\ket{\Psi_2} = \ket{1}\ket{0-1}\ket{0},~\ket{\Psi_3} = \ket{0-1}\ket{0}\ket{1},~\ket{\Psi_4} = \ket{0+1}\ket{0+1}\ket{0+1}.
\end{array}
\end{equation}
In the above equation, we avoid the normalizing coefficients. The above UPB also has the size $(\mathcal{D}-4)$, here $\mathcal{D}$ is eight. Due to the above UPB, it is possible to get a three-qubit BE state of rank-$4$, which is a normalized projector onto a CES. We define the state as $\rho_{edge}$ = $(1/4)(\mathbb{I}-\sum_{i=1}^4\ket{\widetilde{\Psi_i}}\bra{\widetilde{\Psi_i}})$, where $\mathbb{I}$ is an identity operator, acting on a three-qubit Hilbert space, and $\ket{\widetilde{\Psi_i}}$ is the normalized version of $\ket{\Psi_i}$, $\forall i = 1,\dots,4$. We now present the following theorem:

\begin{theorem}\label{theo1}
Starting from the above UPB, it is possible to construct three-qubit noisy bound entangled states, having ranks five to full rank, which satisfy the range criterion.
\end{theorem}

\begin{proof} 
In order to construct noisy BE states, we consider a fully separable state (noise) $\rho_{sep}$, which is produced by taking any convex combination of the states $\ket{\widetilde{\Psi_i}}$. Thus, $\rho_{sep}$ can have rank-$1$ (when $\rho_{sep}$ is any state of the Shift UPB) to full rank (when $\rho_{sep}$ is produced by taking convex combination of all four states of the Shift UPB). Now, a multipartite noisy BE state is a BE state which can be written as a convex combination of an edge state and a fully separable state (noise). So, we define a class of states, $\varsigma(\lambda)$ = $\lambda\rho_{sep} + (1-\lambda)\rho_{edge}$, $\lambda$ is a small positive number. The inseparability of the states $\varsigma(\lambda)$ can be proved by using a particular witness operator given in Ref.~\cite{Bandyopadhyay05}. The structure of those witness operators are given by $\mathcal{W}$ = $\Pi-\gamma\mathbb{I}$, where $\Pi$ is a projection operator onto the subspace spanned by the states of the Shift UPB, the value of $\gamma$ is fixed by minimizing it over all three-qubit fully separable states (see Ref.~\cite{Bandyopadhyay05, Bandyopadhyay08, Terhal01-1} for a detailed definition), and $\mathbb{I}$ is the identity operator, acting on the three-qubit Hilbert space. For the inseparability of the present noisy BE states, we have to take $0<\lambda<\gamma$. In this way, we construct three-qubit noisy BE states, having rank-$5$ to full rank. This is also the widest variety of ranks, a noisy BE state can have, because any rank-$4$ (minimum rank) BE state is supported in a CES \cite{Chen13}.  

The next important question, we want to explore is whether the noisy BE states, constructed here, satisfy the range criterion. Remember that the present noisy BE states are due to a real UPB, so, they must be invariant under PT with respect to every bipartition. Therefore, to prove that a noisy BE state of the present kind satisfies the range criterion, it is sufficient to show that there exist a set of pure product states (here in this section such states are separable across every bipartition) which span the range of the considered noisy BE state. Note that satisfaction of range criterion is only a necessary condition for separability but not a sufficient condition. To examine the states $\varsigma(\lambda)$, we first consider the rank of $\rho_{sep}$, which is $\leq4$. Then the noisy states have rank five to full rank and the rest of the dimension of the three-qubit Hilbert space (where $\varsigma(\lambda)$ is not supported) is spanned by three or less number of orthogonal pure states which are separable across every bipartition. Again, it is known that any three (or less) pure fully separable multipartite states are extendible to a complete basis in any Hilbert space \cite{Divincenzo03}, implying the fact, one can get sufficient product states in the range of $\varsigma(\lambda)$. Thus, the states $\varsigma(\lambda)$ satisfy the range criterion. 
\end{proof}

However, one can also apply the technique given in Ref.~\cite{Halder19-1}, considering real UPBs, to produce noisy BE states, having rank-$5$ to full rank (ranks of the widest variety), which satisfy the range criterion. This is important particularly, when the number of parties is more than three or the dimensions of the subsystems are greater than two. But in that case, the real UPBs should have the property--if the stopper is removed than the rest of a UPB can be extended to a full basis.  

Like Proposition \ref{prop4}, for multipartite systems also, one can show that a particular edge state can be due to a reducible UPB and also due to an irreducible UPB when $\mathcal{D}>8$. Hence, we present the following proposition: 

\begin{proposition}\label{prop9}
In a multipartite system, the total dimension of which $>8$, consider an edge state which is due to a reducible UPB. The same edge state can also be due to an irreducible UPB. Thus, looking at the edge state, it may not be possible to conclude if the state is due to a reducible UPB or it is due to an irreducible UPB.
\end{proposition}

\begin{proof}
The above can be realized via simple constructions. In $2\otimes2\otimes3$, we consider a reducible UPB first. The states of which are given by-
\begin{equation}\label{eq9}
\begin{array}{l}
\ket{0}\ket{1}\ket{0-1},~\ket{1}\ket{0-1}\ket{0},~\ket{0-1}\ket{0}\ket{1},~\ket{0}\ket{0\pm1}\ket{2},~\ket{1}\ket{0\pm1}\ket{2},~\ket{0+1}\ket{0+1}\ket{0+1}.
\end{array}
\end{equation}
Notice that the above UPB is constructed by extending the Shift UPB to a higher dimension. Clearly, the party who is holding the three-level quantum system (say, Charlie) can perform a simple projective measurement to eliminate the states $\ket{0}\ket{0\pm1}\ket{2}$, $\ket{1}\ket{0\pm1}\ket{2}$, keeping the remaining states orthogonal. Corresponding projectors are $\ket{2}\bra{2}$, $\mathbb{I}-\ket{2}\bra{2}$, $\mathbb{I}$ is the identity operator acting on the three-level quantum system. We next construct an irreducible UPB in the same Hilbert space. For this purpose, we consider the states $\ket{1}\ket{0-1}\ket{0}$, $\ket{1}\ket{0-1}\ket{2}$ and taking linear combinations we produce the states $\ket{1}\ket{0-1}\ket{0\pm2}$. The new UPB is given by-
\begin{equation}\label{eq10}
\begin{array}{l}
\ket{0}\ket{1}\ket{0-1},~\ket{1}\ket{0-1}\ket{0\pm2},~\ket{0-1}\ket{0}\ket{1},~\ket{0}\ket{0\pm1}\ket{2},~\ket{1}\ket{0+1}\ket{2},~\ket{0+1}\ket{0+1}\ket{0+1}.
\end{array}
\end{equation}
In the above, because of the new twisted (on Charlie's side) states $\ket{1}\ket{0-1}\ket{0\pm2}$, Charlie is no longer be able to eliminate any state via orthogonality-preserving LOCC. For more details see the Appendix section. Obviously, the states of both the UPBs span the same subspace because the above UPBs have only one difference and that is, in case of the former UPB the states $\ket{1}\ket{0-1}\ket{0}$, $\ket{1}\ket{0-1}\ket{2}$ are used while in the later UPB the state $\ket{1}\ket{0-1}\ket{0\pm2}$ are used. Again, the states $\ket{1}\ket{0-1}\ket{0\pm2}$ are produced taking the linear combinations of the states $\ket{1}\ket{0-1}\ket{0}$, $\ket{1}\ket{0-1}\ket{2}$. As a result of which both the UPBs produce the same edge state.
\end{proof}

Notice that both the UPBs, given above, can be considered as real UPBs of maximum size and thus, they can be useful to produce noisy BE states, having ranks of the widest variety, which satisfy the range criterion. However, for details regarding the notions--reducibility and irreducibility under orthogonality-preserving LOCC, one can go through the papers \cite{Halder18, Halder19, Zhang19} and the references therein. 

We like to mention here that in Ref.~\cite{Walgate02}, an efficient technique was introduced to check if it is possible for any party to begin with a nontrivial and orthogonality-preserving measurement in order to distinguish a given set of orthogonal states perfectly. Thereafter, in many papers this technique has been used to extract different results. For example, in Ref.~\cite{Halder18}, a class of completable sets of pure orthogonal multipartite product states are introduced from which no state can be eliminated via orthogonality-preserving measurements. Based on the same technique, UPBs are classified in the present work and it is shown that looking at the edge state it may not be possible to infer if it is due to a reducible UPB or due to an irreducible UPBs. We also mention that a class of multipartite bound entangled states was introduced in Ref.~\cite{Halder18} but those bound entangled states are supported in completely entangled subspaces. So, they must violate the range criterion. On the other hand, in the present work we talk about multipartite noisy bound entangled states, having ranks of the widest variety, which satisfy the range criterion.

\section{Conclusion}\label{sec5}
In this work we have considered UPBs of size $(\mathcal{D}-4)$ for both bipartite and multipartite systems, where $\mathcal{D}$ is the total dimension of the corresponding Hilbert space. We have argued that a UPB of the above size is important as they might be useful to generate noisy BE states, having ranks of the widest variety, which satisfy the range criterion. Both reducible and irreducible UPBs are considered in this work and it has been shown that a particular edge state can be due to a reducible UPB and also due to an irreducible UPB, where both the UPBs span the same subspace. We have also given a class of reducible UPBs in $d_1\otimes d_2$, having size ($d_1d_2-4$). Furthermore, we have provided certain irreducible UPBs. A class of bipartite rank-$4$ BE states is also constructed and a few properties of these states are discussed. For bipartite systems, a few properties of noisy BE states are studied further. For future studies, one can consider to develop algorithms to provide irreducible UPBs in different composite Hilbert spaces. 

\section*{Appendix}
We first discuss the technique (in a little more detailed way) to prove the irreducibility of the UPB which is given in Eq.~(\ref{eq4}). We imagine that Alice starts the measurement process. An orthogonality-preserving measurement by Alice can be described by a set of measurement operators $\{E_i^a\}_i$ while $\forall i$ the elements $(E_i^a)^\dagger E_i^a$ are positive Hermitian $3\times3$ matrices. For any value of $i$, we take $(E_i^a)^\dagger E_i^a$ = $[e_{kk^\prime}]$, where $e_{kk^\prime}$ are matrix entries and $k, k^\prime$ = $0,1,2$. We assume that the measurement is orthogonality preserving. So, the post-measurement states remain orthogonal. The post-measurement states corresponding to the states $\ket{0}\ket{0-1}$, $\ket{1}\ket{1-2}$ are $(E_i^a\otimes\mathbb{I})\ket{0}\ket{0-1}$ and $(E_i^a\otimes\mathbb{I})\ket{1}\ket{1-2}$ (normalization is ignored here), $\mathbb{I}$ is the identity operator acting on Bob's Hilbert space. Taking inner product of these two post-measurement states, we get $\bra{0}(E_i^a)^\dagger E_i^a\ket{1}\bra{0-1}\ket{1-2}$ = $0$, i.e., $e_{01}$ = $e_{10}$ = $0$. Similarly, if we take the post-measurement states corresponding to the states $\ket{0}\ket{0-1}$ and $\ket{1-2}\ket{0}$ then after taking the inner product, we get $e_{01} - e_{02}$ = $0$. But $e_{01}$ is already zero and thus, $e_{02}$ = $e_{20}$ = $0$. Following the same technique and taking different combinations of states, it is possible to show that all off-diagonal entries $e_{kk^\prime}$ = $0$, $k, k^\prime$ = $0,1,2$, $k\neq k^\prime$. Again, we take the post-measurement states corresponding to the states $\ket{0-1}\ket{3}$ and $\ket{0+1+2}\ket{0+1+2+3}$ and considering the inner product, we get $e_{00} + e_{01} + e_{02}$ = $e_{10} + e_{11} + e_{12}$. But it is already proved that the off-diagonal entries are zero. So, we get $e_{00}$ = $e_{11}$. Taking other combinations of states, we can show that the diagonal entries are all equal. Clearly, the elements $(E_i^a)^\dagger E_i^a$ are proportional to a $3\times3$ identity matrix, implying the fact that Alice is not able to begin with a nontrivial orthogonality-preserving measurement. This is also sufficient to establish that Alice cannot eliminate any state via orthogonality-preserving LOCC. In the same way, as given here, one can easily prove that Bob also cannot eliminate any state via orthogonality-preserving LOCC. All these arguments, given one by one in the above, complete the proof of irreducibility of the UPB of Eq.~(\ref{eq4}). We now mention an important fact. Even if a given UPB in $\mathbb{C}^3\otimes\mathbb{C}^4$ is reducible, the party whose system dimension is three, cannot eliminate any state via orthogonality-preserving LOCC. Because the reducible UPB in $\mathbb{C}^3\otimes\mathbb{C}^4$, anyway, contains an irreducible UPB in a $\mathbb{C}^3\otimes\mathbb{C}^3$ subspace. However, to prove the irreducibility of the UPB which is given in Eq.~(\ref{eq10}), one can go through the same process which is followed in the above. In fact, we also mention that for any reducible UPB in $\mathbb{C}^2\otimes\mathbb{C}^2\otimes\mathbb{C}^3$, any party whose system dimension is two, cannot eliminate any state via orthogonality-preserving LOCC. Because the reducible UPB in $\mathbb{C}^2\otimes\mathbb{C}^2\otimes\mathbb{C}^3$, anyway, contains an irreducible UPB in a $\mathbb{C}^2\otimes\mathbb{C}^2\otimes\mathbb{C}^2$ subspace.

\section*{References}
\bibliographystyle{elsarticle-num}
\bibliography{ref}
\end{document}